\documentclass[a4paper,UKenglish,cleveref, autoref, thm-restate]{lipics-v2021}
\pdfoutput=1 
\hideLIPIcs  
\usepackage[T1]{fontenc}
\usepackage{thmtools}
\usepackage{thm-restate}
\usepackage{enumerate}
\usepackage{ascmac}
\usepackage{comment}
\usepackage{color}
\usepackage{thm-restate}
\usepackage{algorithm}
\usepackage{algpseudocode}

\newcommand{\LBAE}{\textsf{LBAE}}
\newcommand{\SCAE}{\textsf{SCAE}}
\newcommand{\occ}{\mathit{occ}}
\newcommand{\LCP}{\mathit{LCP}}
\newcommand{\lcp}{\mathit{lcp}}
\newcommand{\lce}{\mathit{lce}}
\newcommand{\Z}{\mathsf{Z}}
\newcommand{\B}{\mathsf{B}}
\newcommand{\BG}{\mathsf{BG}}
\newcommand{\period}{\mathsf{per}}
\newcommand{\border}{\mathsf{bord}}
\newcommand{\cover}{\mathsf{cov}}
\newcommand{\range}{\mathsf{range}}
\newcommand{\C}{\mathsf{C}}
\newcommand{\R}{\mathsf{R}}
\newcommand{\Rstar}{\mathsf{R}^\star}
\newcommand{\idx}{\mathit{idx}}
\newcommand{\clen}{\mathit{clen}}

\title{Shortest cover after edit} 

\author{Kazuki~Mitani}{Graduate School of Information Science and Technology, Hokkaido University, Japan}{kazukida199911204649@eis.hokudai.ac.jp}{}{}
\author{Takuya~Mieno}{Department of Computer and Network Engineering, University of Electro-Communications, Japan}{tmieno@uec.ac.jp}{https://orcid.org/0000-0003-2922-9434}{}
\author{Kazuhisa~Seto}{Faculty of Information Science and Technology, Hokkaido University, Japan}{seto@ist.hokudai.ac.jp}{https://orcid.org/0000-0001-9043-7019}{}
\author{Takashi~Horiyama}{Faculty of Information Science and Technology, Hokkaido University, Japan}{horiyama@ist.hokudai.ac.jp}{https://orcid.org/0000-0001-9451-259X}{}
\authorrunning{K.~Mitani~et al.} 

\Copyright{Kazuki Mitani, Takuya Mieno, Kazuhisa Seto, and Takashi Horiyama} 

\ccsdesc[500]{Mathematics of computing~Combinatorial algorithms}

\keywords{string algorithm, border, cover, quasi-periodicity, dynamic string} 

\category{} 

\relatedversion{}

\nolinenumbers 

\begin{document}

\maketitle

\begin{abstract}
  This paper investigates the (quasi-)periodicity of a string when the string is edited.
  A string $C$ is called a cover (as known as a quasi-period) of a string $T$
  if each character of $T$ lies within some occurrence of $C$.
  By definition, a cover of $T$ must be a border of $T$; that is,
  it occurs both as a prefix and as a suffix of $T$.
  In this paper, we focus on the changes in the longest border and the shortest cover of a string when the string is edited only once.
  We propose a data structure of size $O(n)$
  that computes the longest border and the shortest cover of the string in $O(\ell \log n)$ time
  after an edit operation (either insertion, deletion, or substitution of some string) is applied to the input string $T$ of length $n$,
  where $\ell$ is the length of the string being inserted or substituted.
  The data structure can be constructed in $O(n)$ time given string $T$.
\end{abstract}

\section{Introduction}\label{sec:intro}

Periodicity and repetitive structure in strings are important concepts in the field of stringology
and have applications in various areas, such as pattern matching and data compression.
A string $u$ is called a \emph{period-string} (or simply a \emph{period}) of string $T$ if $T = u^ku'$ holds for some positive integer $k$ and some prefix $u'$ of $u$.
While periods accurately capture the repetitive structure of strings, the definition is too restrictive.
In contrast, alternative concepts that capture a sort of periodicity with relaxed conditions have been studied.
A \emph{cover}~(a.k.a.~\emph{quasi-period}) of a string is a typical example of such a concept~\cite{ApostolicoE90,ApostolicoE93}.
A string $v$ is called a cover of $T$ if every character in $T$ lies within some occurrence of $v$.
In other words, $T$ can be written as a repetition of occurrences of $v$ that are allowed to overlap.
By definition, a cover of $T$ must occur as both a prefix and a suffix of $T$, and such string is called a \emph{border} of $T$.
Therefore, a cover of $T$ is necessarily a border of $T$.
For instance, $v = \mathtt{aba}$ is a cover for $S = \mathtt{abaababa}$, and $v$ is both a prefix and a suffix of $T$.
Then, the string $v = \mathtt{aba}$ of length $3$ can be regarded as an ``almost'' period-string in $S$
while the shortest period-string of $S$ is $\mathtt{abaab}$ of length $5$.
Thus, covers can potentially discover quasi-repetitive structures not captured by periods.
The concept of covers (initially termed quasi-periods) was introduced by Apostolico and Ehrenfeucht~\cite{ApostolicoE90,ApostolicoE93}.
Subsequently, an algorithm to compute the shortest cover offline in linear time was proposed by Apostolico et al.~\cite{ApostolicoFI91}.
Furthermore, an online and linear-time method was presented by Breslauer~\cite{Breslauer92}.
Gawrychowski et al. explored cover computations in streaming models~\cite{GawrychowskiRS19}.
In their problem setting, the computational complexity is stochastic.
Other related work on covers can be found in the survey paper by Mhaskar and Smyth~\cite{MhaskarS22}.

In this paper, we investigate the changes in the shortest cover of a string $T$ when $T$ is edited and design algorithms to compute it.
As mentioned above, the shortest cover of $T$ is necessarily a border of $T$, so we first consider how to compute borders when $T$ is edited.
To the best of our knowledge,
there is only one explicitly-stated result on the computation of borders in a dynamic setting:
the longest border of a string $S$
(equivalently, the smallest period of $S$) can be maintained in $O(|S|^{o(1)})$ time
per character substitution operation (Corollary 19 of~\cite{AmirBCK19}).
Also, although is not stated explicitly,
an $O(\log^3 n)$-time (w.h.p.) algorithm can be obtained by using the results on the \emph{PILLAR model} in dynamic strings~\cite{PILLAR_dynamic}.
We are unsure
whether their results can be applied to compute the shortest cover in a dynamic string.
Instead, we focus on studying the changes in covers when a \emph{factor} is edited only once.
We believe that this work will be the first step towards the computation of covers for a fully-dynamic string.
We now introduce two problems: the {\LBAE}~(longest border after-edit) query and the {\SCAE}~(shortest cover after-edit) query for the input string $T$ of length $n$.
The {\LBAE} query (resp., the {\SCAE} query) is, given an edit operation \emph{on the original string $T$} as a query,
to compute the longest border (resp., the shortest cover) of the edited string.
We note that, after we answer a query, the edit operation is discarded.
That is, the following edit operations are also applied to the original string $T$.
This type of problem is called the \emph{after-edit model}~\cite{Amir2017}.
Also, in our problems, the edit operation includes insertion, deletion, or substitution of strings of length one or more.
Our main contribution is designing an $O(n)$-size data structure that can answer both {\LBAE} and {\SCAE} queries
in $O(\ell \log n)$ time, where $\ell$ is the length of the string being inserted or substituted.
The data structures can be constructed in $O(n)$ time.

\paragraph*{Related Work on After-Edit Model.}
The after-edit model was formulated by Amir et al.~\cite{Amir2017}.
They proposed an algorithm to compute the \emph{longest common factor} (LCF) of two strings in the after-edit model.
This problem allows editing operations on only one of the two strings.
Abedin et al.~\cite{AbedinH0T18} subsequently improved their results.
Later, Amir et al.~\cite{AmirCPR20} generalized this problem
to a \emph{fully-dynamic} model and proposed an algorithm that maintains the LCF
in $\tilde{O}(n^{\frac{2}{3}})$ time\footnote{The $\tilde{O}(\cdot)$ notation hides poly-logarithmic factors.} per edit operation.
Charalampopoulos et al.~\cite{Charalampopoulos20} improved the maintenance time
to amortized $\tilde{O}(1)$ time with high probability per substitution operation.
Urabe et al.~\cite{UrabeNIBT18} addressed the problem of computing the \emph{longest Lyndon factor} (LLF) of a string in the after-edit model.
The insights gained from their work were later applied to solve the problem
of computing the LLF of a fully-dynamic string~\cite{AmirCPR20}.
Problems of computing the \emph{longest palindromic factor} and \emph{unique palindromic factors} in a string
were also considered in the after-edit model~\cite{FunakoshiNIBT21,FunakoshiM21,Mieno2023}.
\section{Preliminaries}\label{sec:pre}
\subsection{Basic Definitions and Notations}
\noindent\textbf{Strings.}\quad
Let $\Sigma$ be an \emph{alphabet}.
An element in $\Sigma$ is called a \emph{character}.
An element in $\Sigma^\star$ is called a \emph{string}.
The length of a string $S$ is denoted by $|S|$.
The string of length $0$ is called the \emph{empty string}
and is denoted by $\varepsilon$.
If a string $S$ can be written as a concatenation of
three strings $p, f$ and $s$,
i.e., $S = pfs$,
then $p, f$ and $s$ are called
a \emph{prefix}, a \emph{factor}, and a \emph{suffix} of $S$,
respectively.
Also, if $|p| < |S|$ holds, $p$ is called a \emph{proper} prefix of $S$.
Similarly, $s$ is called a proper suffix of $S$ if $|s| < |S|$ holds.
For any integer $i,j$ with $1\le i \le j \le |S|$,
we denote by $T[i]$ the $i$-th character of $S$,
and by $T[i.. j]$ the factor of $S$ starting at position $i$ and ending at position $j$.
For convenience, let $T[i'..j'] = \varepsilon$ for any $i', j'$ with $i' > j'$.
For two strings $S$ and $T$,
we denote by $\LCP(S, T)$ the \emph{longest common prefix} of $S$ and $T$.
Also, we denote by $\lcp(S, T)$ the length of $\LCP(S, T)$.
If $f = S[i.. i+|f|-1]$ holds, we say that $f$ \emph{occurs} at position $i$ in $S$.
Let $\occ_S(f) = \{i\mid f = S[i.. i+|f|-1]\}$ be the set of occurrences of $f$ in $S$.
Further let $\mathit{cover}_S(f) = \{p\mid p\in [i, i+|f|-1]\text{ for some }i\in\occ_S(f)\}$
be the set of positions in $S$ that are covered by some occurrence of $f$ in $S$.
A string $f$ is called a \emph{cover} of $S$ 
if $\mathit{cover}_S(f) = \{1,\dots , |S|\}$ holds.
A string $b$ is called a \emph{border} of a non-empty string $S$
if $b$ is both a proper prefix of $S$ and a proper suffix of $S$.
We say that $S$ has a border $b$ when $b$ is a border of $S$.
By definition, any non-empty string has a border $\varepsilon$.
If a string $S$ has a border $b$, integer $p = |S|-|b|$ is called a \emph{period} of $S$.
We sometimes call 
the smallest period of $S$ \emph{the} period of $S$.
Similarly, we call the longest border of $S$ \emph{the} border of $S$, and
the shortest cover of $S$ \emph{the} cover of $S$.
Also, we denote by $\period(S)$, $\border(S)$, and $\cover(S)$
the period of $S$, the border of $S$, and the cover of $S$, respectively.
The rational number $|S|/\period(S)$ is called the \emph{exponent} of $S$.
We say that $S$ is \emph{periodic} if $\period(S) \le |S|/2$.
A string $S$ is said to be \emph{superprimitive} if $\cover(S) = S$.

\noindent\textbf{After-edit Model.}\quad
The \emph{after-edit model} is,
given an edit operation on the input string $T$ as a query,
to compute the desired objects on the edited string $T'$
that is obtained by applying the edit operation to $T$.
Note that in the after-edit model,
each query, namely each edit operation, is discarded
after we finish computing the desired objects on $T'$,
so the next edit operation will be applied to the original string $T$.
In this paper, edit operations consist of
inserting a string
and substituting a factor with another string.
Note that factor substitutions contain factor deletions
since substituting a factor with the empty string $\varepsilon$
is identical to deleting the factor.
We denote an edit operation as $\phi(i, j, w)$
where $1 \le j \le |T|$, $1\le i \le j+1$ and $w \in \Sigma^\star$:
if $i \le j$, $\phi(i, j, w)$ means to substitute $T[i.. j]$ for $w$.
If $i = j + 1$, $\phi(i, j, w)$ means to insert $w$ just after $T[i-1]$.
In both cases, the resulting string is $T' = T[1.. i-1]wT[j+1..|T|]$
and thus $T'[i.. i+|w|-1] = w$.
For a given query $\phi(i, j, w)$, let $L_{i,j} = T[1.. i-1]$ and $R_{i,j} = T[j+1.. |T|]$.
We will omit the subscripts when they are clear from the context.
Thus, $T' = LwR$.
We consider the two following problems with the after-edit model:
\begin{itembox}[l]{\bf {\LBAE} (Longest Border After-Edit) query}
  \begin{description}
    \item[Preprocess:] A string $T$ of length $n$.
    \item[Query:] An edit-operation $\phi(i, j, w)$.
    \item[Output:] The longest border of $T' = L_{i,j}wR_{i,j}$.
  \end{description}
\end{itembox}

\begin{itembox}[l]{\bf {\SCAE} (Shortest Cover After-Edit) query}
  \begin{description}
    \item[Preprocess:] A string $T$ of length $n$.
    \item[Query:] An edit-operation $\phi(i, j, w)$.
    \item[Output:] The shortest cover of $T' = L_{i,j}wR_{i,j}$.
  \end{description}
\end{itembox}
In the following, we fix the input string $T$ of arbitrary length $n > 0$.
Also, we assume that the computation model in this paper is the word-RAM model with word size $\Omega(\log n)$.
We further assume that the alphabet $\Sigma$ is \emph{linearly-sortable},
i.e., we can sort $n$ characters from the input string in $O(n)$ time.

\subsection{Combinatorial Properties of Borders and Covers}\label{subsec:border}
This subsection describes known properties of borders and covers.
\paragraph*{Periodicity of Borders}
Let us consider partitioning the set $\mathcal{B}_T$ of borders of a string $T$ of length $n$.
Let $G_1, G_2, \ldots, G_m$ be the sets of borders of $T$
such that $\{G_1, G_2, \ldots, G_m\}$ is a partition of $\mathcal{B}_T$ and
for each set, all borders in the same set have the same smallest period.
Let $p_k$ be the period of borders belonging to $G_k$.
Without loss of generality, we assume that they are indexed so that $p_k > p_{k+1}$ for every $1\le k < m$.
We call $G_k$ the \emph{$k$-th group}.
Then, the following fact is known:
\begin{proposition}[\cite{KnuthMP77,KarpinskiRS97}]\label{prop:group}
  For the partition $\{G_1, G_2, \ldots, G_m\}$ of $\mathcal{B}_T$ defined above,
  the following statements hold:
  \begin{enumerate}
    \item For each $1 \le k \le m$, 
      the lengths of borders in the $k$-th group can be represented as
      a single arithmetic progression with common difference $p_k$.
    \item If a group contains at least three elements,
      the borders in the group except for the shortest one are guaranteed to be periodic.
    \item The number $m$ of sets is in $O(\log n)$.
  \end{enumerate}
\end{proposition}

\paragraph*{Properties of Covers}
The following lemma summarizes 
some basic properties of covers, which we will use later.

\begin{lemma}[\cite{Breslauer92,MooreS94}]\label{lem:cover_properties}
  For any string $T$, the following statements hold.
  \begin{enumerate}
    \item The cover $\cover(T)$ of $T$ is either $\cover(\border(T))$ or $T$ itself.
    \item $\cover(T)$ is superprimitive and non-periodic.
    \item 
      Let $v$ be a cover of $T$, and $u$ be a factor of $T$ which is shorter than $v$.
      Then $u$ is a cover of $T$ if and only if $u$ is a cover of $v$.
  \end{enumerate}
\end{lemma}

\subsection{Algorithmic Tools}
This subsection shows algorithmic tools we will use later.

\paragraph*{Border Array and Border-group Array.}
The \emph{border array} $\B_T$ of a string $T$ is an array of length $n$
such that $\B_T[i]$ stores the length of the border of $T[1..i]$ for each $1 \le i \le n$.
Also, for convenience, let $\B_T[0] = 0$ for any string $T$.
There is a well-known \emph{online} algorithm for linear-time computation of the border array~(e.g., see~\cite{Gusfield97}).
While the worst-case running time of the algorithm is $O(n)$ per a character,
it can be made $O(\log n)$ by constructing $\B$ with the \emph{strict border array} proposed in~\cite{KnuthMP77}.
\begin{lemma}[\cite{Gusfield97,KnuthMP77}]\label{lem:border}
  For each $1 \le i \le n$,
  if we have $T[1.. i-1]$ and $\B_{T[1..i-1]}$,
  then we can compute $\B_{T[1..i]}$ in
  worst-case $O(\log n)$ time and amortized $O(1)$ time
  given the next character $T[i]$.
\end{lemma}
Assume that we already have a string $T$ of length $n$ and its border array $\B_T$.
Then, given a prefix $L$ of $T$ and any string $w$ of length $\ell$,
let $\beta(n, \ell)$ be the computation time to obtain the border array of $Lw$ from $\B_T$.
Now $\beta(n, \ell) \in O(\ell\log n)$ holds due to Lemma~\ref{lem:border}, 

Next, we introduce a data structure closely related to the border array.
The \emph{border-group array} $\BG_T$ of a string $T$ is an array of length $n$
such that, for each $1 \le i \le n$,
$\BG_T[i]$ stores the length of the shortest border of $T[1..i]$ whose smallest period equals $\period(T[1..i])$
if such a border exists, and $\BG_T[i] = i$ otherwise.
By definition,
if $T[1.. i]$ is a border of $T$ and it belongs to group $G_k$,
then $\BG_T[i]$ stores the first (the smallest) term of the arithmetic progression
representing the lengths of borders in $G_k$.
This is why we named $\BG_T$ the border-group array.
Also, the common difference $p_k = i - \B_T[i]$ can be obtained from $\B_T$ if $\BG_T[i] \ne i$.
See Figure~\ref{fig:BG} for an example.
\begin{figure}[tb]
  \center{\includegraphics[width=1.0\linewidth]{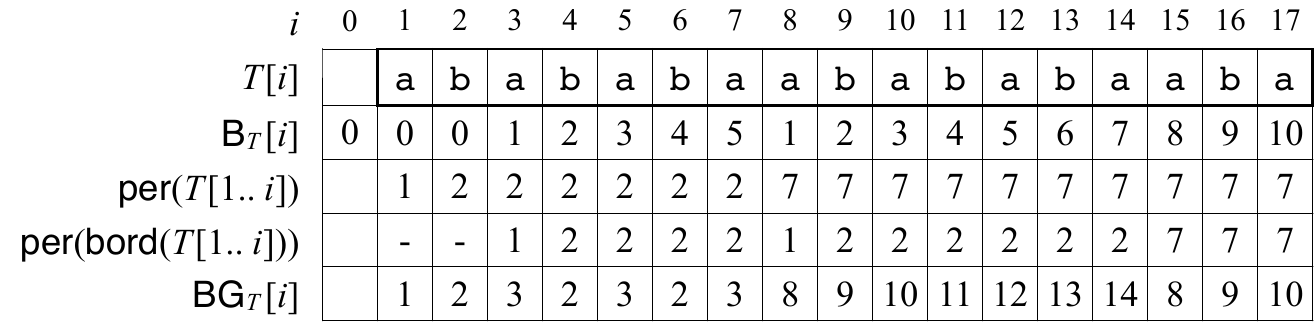}}
  \caption{An example of a border array and a border-group array.
    For position $i = 7$, the period of $T[1.. 7]$ is $2$.
    All borders of $T[1.. 7]$ are 
    $\mathtt{ababa}$, $\mathtt{aba}$, $\mathtt{a}$, and $\varepsilon$.
    Also, their smallest periods are $2$, $2$, $1$, and $0$, respectively.
    Thus $\BG_{T}[7] = |\mathtt{aba}| = 3$.
    For position $i = 8$, the period of $T[1.. 8]$ is $7$.
    Any border of $T[1.. 8]$ does not have period $7$, and thus $\BG_{T}[8] = i = 8$.
  } \label{fig:BG}
\end{figure}
We can compute the border-group array in linear time in an online manner together with $\B_T$.
Before proving it, we note a fact about periods.
\begin{proposition}\label{prop:period}
  If $u$ is a factor of $v$, then $\period(u) \le \period(v)$.
\end{proposition}
\begin{lemma}\label{lem:bordergroup}
  For each $1 < i \le n$,
  if we have $T[1.. i-1]$, $\BG_{T[1..i-1]}$, and $\B_{T[1.. i]}$,
  then we can compute $\BG_{T[1..i]}$ in $O(1)$ time
  given the next character $T[i]$.
\end{lemma}
\begin{proof}
  By definition, $\BG_{T[1..1]} = [1]$.
  Now let $p = i - \B_{T[1.. i]}[i]$ and $q = \B_{T[1.. i]}[i] - \B_{T[1.. i]}[\B_{T[1.. i]}[i]]$,
  meaning that $p = \period(T[1..i])$ and $q = \period(\border(T[1..i]))$.
  If $p = q$, then we set $\BG_{T[1..i]}[i] = \BG_{T[1..i-1]}[\B_{T[1..i]}[i]]$ since
  $\period(T[1.. i]) = \period(\border(T[1..i]))$.
  Otherwise, $\period(T[1.. i]) > \period(\border(T[1..i]))$,
  and thus, the period of \emph{any} border of $T[1.. i]$ is smaller than $\period(T[1..i])$
  by Proposition~\ref{prop:period}.
  Hence we set $\BG_{T[1..i]}[i] = i$.
  The running time of the algorithm is $O(1)$.
\end{proof}

\paragraph*{Longest Common Extension Query.}
The \emph{longest common extension} query (in short, \emph{LCE} query) is,
given positions $i$ and $j$ within $T$,
to compute $\lcp(T[i.. |T|], T[j.. |T|])$.
We denote the answer of the query as $\lce_T(i, j)$.
We heavily use the following result in our algorithms.
\begin{lemma}[E.g.,~\cite{BenderF00}]\label{lem:lce}
  We can answer any LCE query in $O(1)$ time after $O(n)$-time and space preprocessing on the input string $T$.
\end{lemma}

\paragraph*{Prefix Table.}
The \emph{prefix table} $\Z_S$ of a string $S$ of length $m$ is an array of length $m$
such that $\Z_S[i] = \LCP(S, S[i..m])$ for each $1 \le i \le m$.
\begin{lemma}[\cite{MainL84,Gusfield97}]
  Given a string $S$ of length $m$ over a general unordered alphabet,
  we can compute the prefix table $\Z_S$ in $O(m)$ time\footnote{The algorithm described in~\cite{Gusfield97} is known as \emph{Z-algorithm}, so we use $\Z$ to represent the prefix table.}.
\end{lemma}
We emphasize that this linear-time algorithm does not require linearly-sortability of the alphabet.

\paragraph*{Internal Pattern Matching.}
The \emph{internal pattern matching} query (in short, \emph{IPM} query) is,
given two factors $u, v$ of $T$ with $|v| \le 2|u|$,
to compute the occurrences of $u$ in $v$.
The output is represented as an arithmetic progression due to the lengths constraint and periodicity~\cite{KociumakaRRW13}.
If $u$ occurs in $v$,
we denote by $\mathsf{rightend}(u, v)$ the ending position of the rightmost occurrence of $u$ in $v$.
\begin{lemma}[\cite{KociumakaRRW13}]\label{lem:IPM}
  We can answer any IPM query in $O(1)$ time after $O(n)$-time and space preprocessing on the input string $T$,
\end{lemma}
\section{Longest Border After Edit} \label{sec:border}

This section proposes an algorithm to solve the {\LBAE} problem.
In the following, we assume that $|L| \ge |R|$ and $|w| \le |L|/2$ for a fixed query $\phi(i, j, w)$.
Because,
when $|L| < |R|$,
running our algorithm on the reversal inputs can answer {\LBAE} queries without growing complexities.
Also, if $|w| > |L|/2$, then $|w| > |T'|/5$ holds since $T' = LwR$ and $|L| \ge |R|$.
We can obtain the border of $T'$ in $O(|T'|) = O(|w|)$ time by computing the border array of $T'$ from scratch.

We compute the border of $T'$ in the following two steps.
\textbf{Step 1}: Find the longest border of $T'$ which is longer than $R$.
\textbf{Step 2}: Find the longest border of $T'$ of length at most $|R|$ if nothing is found in step 1.
Step 2 can be done in constant time
by pre-computing all borders of $T$ and the longest border of $T$ of length at most $k$ for each $k$ with $1\le k \le n$.
Thus we focus on step 1, i.e., how to find the longest border of $T'$ which is longer than $R$.
We observe that such a border is the concatenation of some border of $Lw$ and $R$~(see Figure~\ref{fig:borderofLw}).
\begin{figure}[tb]
  \center{\includegraphics[width=0.5\linewidth]{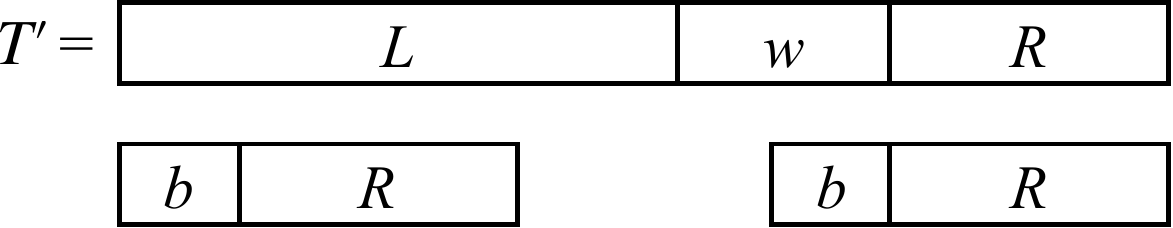}}
  \caption{A border of $T'$ which is longer than $R$ is written as $bR$ where $b$ is a border of $Lw$.
  } \label{fig:borderofLw}
\end{figure}
By pre-computing the border array $\B_T$,
the border array $\B_{Lw}$ of $Lw$ can be computed in $O(\beta(n, |w|))$ time starting from $\B_L = \B_T[1.. |L|]$~(Lemma~\ref{lem:border}).
Let $b_{Lw}$ be the border of $Lw$.
There are two cases:
(i) $|b_{Lw}| \le |w|$ or (ii) $|b_{Lw}| > |w|$.
We call the former case the short border case and the latter case the long border case.

\subsection{Short Border Case} \label{subsec:shortborder}

In this case, the length of the border of $T' = LwR$ is at most $|b_{Lw}R| \le |wR| \le |Lw|$,
so its prefix-occurrence ends within $Lw$.
Also, for any border $b$ of $Lw$, string $bR$ is a border of $T'$ if and only if  $\lce_{T'}(|b|+1, |Lw|+1) = |R|$ holds.
Thus, we pick up each border $b$ of $Lw$ in descending order of length
and check whether $bR$ is a border of $T'$ by computing $\lce_{T'}(|b|+1, |Lw|+1)$.
Since $Lw$ has at most $|w|$ borders, constant-time LCE computation results in a total of $O(|w|)$ time.
If $|b| + |R| \le |L|$ then we can use the LCE data structure on the original string $T$ since $\lce_{T'}(|b|+1, |Lw|+1) = \lce_T(|b|+1, j+1)$ holds.
Otherwise, we may compute the longest common prefix of $w$ and some suffix of $R$, which cannot be computed by applying LCE queries on $T$ na\"ively.
To resolve this issue, we compute the prefix table $\Z_{W}$ of $W$ in $O(|W|) = O(|w|)$ time
where $W = w\cdot R[|R|-|w|+1.. |R|]$ is the concatenation of $w$ and the length-$|w|$ suffix of $R$.
Note that $|R| >  |w|$ holds here since $|R| > |L| - |b| \ge |L| - |w| \ge |w|$ by the assumptions in this case.
Then the longest common prefix of $w$ and any suffix of $R$ of length at most $|w|$ is obtained in constant time, and so is $\lce_{T'}(|b|+1, |Lw|+1)$.
Therefore, we can compute $\lce_{T'}(|b|+1, |Lw|+1)$ for all borders $b$ of $Lw$ in a total of $O(\beta(n, |w|) + |w| + \log n)$ time.

\subsection{Long Border Case}\label{subsec:longborder}

Firstly, we give some observations for the long border case; $|b_{Lw}|  > |w|$.
If $|b_{Lw}|  \le |L|$, then $w = L[|b_{Lw}|  - |w| + 1.. |b_{Lw}| ]$ holds.
If $|b_{Lw}|  > |L|$, then the period $p_{Lw}$ of $Lw$ is $p_{Lw} = |Lw|-|b_{Lw}|  < |Lw| - |L| = |w|$.
Let $k$ be the smallest integer such that $kp_{Lw} \ge |w|$.
Since $k \ge 2$, $kp_{Lw} \le 2(k-1)p_{Lw} < 2|w| \le |L|$ holds.
Thus $w = L[|L|-kp_{Lw} + 1.. |L|-kp_{Lw}+|w|]$ holds~(see also Figure~\ref{fig:w_occursinL}).
\begin{figure}[tb]
  \center{
    \includegraphics[width=\linewidth]{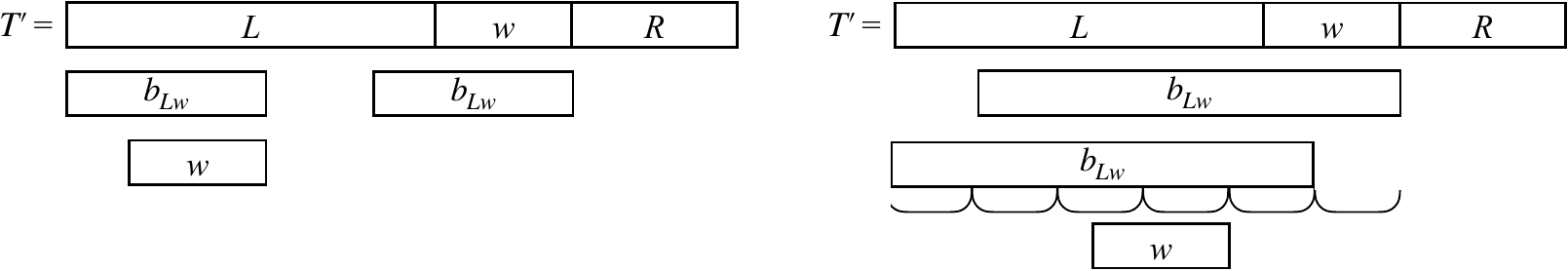}
  }
  \caption{Left: If $b_{Lw}$ is not longer than $L$, then $w$ occurs within $L$ as a suffix of the prefix-occurrence of $b_{Lw}$ in $L$.
    Right: If $b_{Lw}$ is longer than $L$, then $w$ occurs within $L$ because of the periodicity of $b_{Lw}$.
  } \label{fig:w_occursinL}
\end{figure}
Thus, in both cases, $w$ occurs \emph{within} $L$, which is a factor of the original $T$.
From this observation, any single LCE query on $T'$ can be simulated by constant times LCE queries on $T$
because any LCE query on $w = T'[|L|+1.. |L|+|w|]$ can be simulated by a constant number of LCE queries on another occurrence of $w$ within $L$.
Therefore, in the following, we use the fact that any LCE query on $T' = LwR$ can be answered in $O(1)$ time as a black box.

Now, we show some properties of the border of $T'$.
As we mentioned in Proposition~\ref{prop:group},
the sets of borders of $Lw$ can be partitioned into $m \in O(\log n)$ groups w.r.t.~their smallest periods.
Let $G_1, G_2, \ldots, G_m$ be the groups
such that $p_k > p_{k+1}$ for every $1\le k < m$,
where $p_k$ is the period of borders in $G_k$.
Next, let us assume that there exists a border of $T'$ which is longer than $R$.
Let $b^\star$ be the border of $Lw$ such that $b^\star R$ is the border of $T'$.
Further let $k^\star$ be the index of the group to which $b^\star$ belongs.
There are three cases:
(i)   $b^\star$ is periodic and $\period(b^\star) = p_{k^\star} = \period(b^\star R)$,
(ii)  $b^\star$ is periodic and $\period(b^\star) = p_{k^\star} \neq \period(b^\star R)$, or
(iii) $b^\star$ is not periodic.
The first two cases are illustrated in Figure~\ref{fig:longbordercase}.
For the case (i), the following lemma holds.
Here, for a group $G_k$, let $\alpha_k$ be the exponent of the longest prefix of $T'$ with period $p_k$.
\begin{lemma}\label{lem:case1}
  If $b^\star$ is periodic and $p_{k^\star} = \period(b^\star R)$, then
  $|b^\star|\le \alpha_{k^\star}p_{k^\star} - |R|$ holds.
\end{lemma}
\begin{proof}
  Assume the contrary that $|b^\star| > \alpha_{k^\star}p_{k^\star} - |R|$ holds.
  Then $|b^\star R|/p_{k^\star} > \alpha_{k^\star}$ holds.
  This contradicts the maximality of $\alpha_{k^\star}$ since $b^\star R$ occurs as a prefix of $T'$ and $\period(b^\star R) = p_{k^\star}$.
\end{proof}
For the case (ii), the following lemma holds.
Here, for a group $G_k$, let $r_{k} = \lce_{T'}(|T'| - |R| -p_{k}+1, |T'| - |R| + 1)$.
\begin{lemma}\label{lem:case2}
  If $b^\star$ is periodic and $p_{k^\star} \ne \period(b^\star R)$, then
  $|b^\star| = \alpha_{k^\star}p_{k^\star} - r_{k^\star}$ holds.
\end{lemma}
\begin{proof}
  Since the period of $b^\star$ is $p_{k^\star}$, the longest prefix of $T'[|T'|-|b^\star R|+1.. |T'|]$ with period $p_{k^\star}$ is of length $|b^\star| + r_{k^\star}$.
  Thus, by the definition of $\alpha_{k^\star}$, $\alpha_{k^\star}p_{k^\star} = |b^\star| + r_{k^\star}$ holds since $T'[1.. |b^\star R|] = T'[|T'|-|b^\star R|+1.. |T'|]$.
  Therefore, $|b^\star| = \alpha_{k^\star}p_{k^\star} - r_{k^\star}$.
\end{proof}
\begin{figure}[tb]
  \center{
    \includegraphics[width=\linewidth]{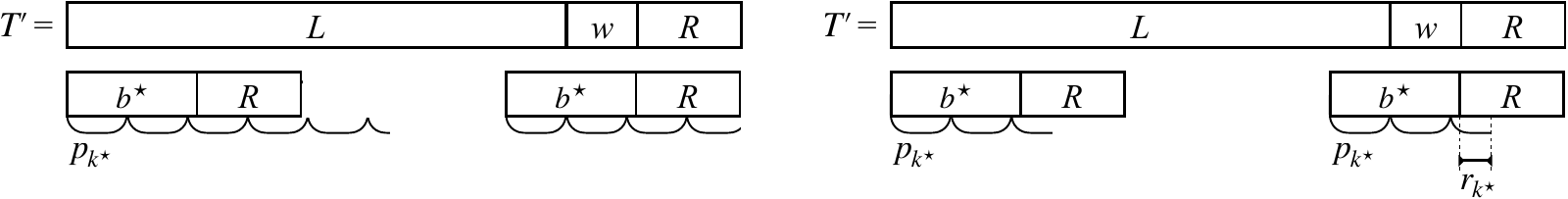}
  }
  \caption{Left:  Illustration for the case (i)  $b^\star$ is periodic and $\period(b^\star) = p_{k^\star} = \period(b^\star R)$.
    The period $p_{k^\star}$ repeats five times and a little more in $T'$.
    Then $|b^\star|$ is at most the length of the repetition minus $|R|$~(Lemma~\ref{lem:case1}).
    Right: Illustration for the case (ii) $b^\star$ is periodic and $\period(b^\star) = p_{k^\star} \neq \period(b^\star R)$.
    Since the maximal repetition of period $p_{k^\star}$ ends within $R$,
    the length $|b|$ is equal to the length of the maximal repetition minus $r_{k^\star}$
    where $r_{k^\star}$ is the length of the suffix of the repetition that enters $R$~(Lemma~\ref{lem:case2}).
  } \label{fig:longbordercase}
\end{figure}
Clearly, if $p_{k^\star} = \period(b^\star R)$, then $r_{k^\star} = |R|$ holds.
Hence, by combining the two above lemmas, we obtain the next corollary:
\begin{corollary}\label{cor:bord}
  If $b^\star$ is periodic, then
  $b^\star$ is the longest border of $Lw$ whose length is at most $\alpha_{k^\star}p_{k^\star} - r_{k^\star}$.
\end{corollary}
Based on this corollary, we design an algorithm to answer the {\LBAE} queries.

\paragraph*{Algorithm.}
The idea of our algorithm is as follows:
given a query, we first initialize \emph{candidates-set} $\mathcal{C} = \emptyset$, which will be a set of candidates for the length of the border of $T'$.
Next, for each group of borders of $Lw$, we calculate a constant number of candidates from the group and add their lengths to the candidates-set $\mathcal{C}$ (the details are described below).
In the end, we choose the maximum from $\mathcal{C}$ and output it.

Now we consider the $k$-th group $G_k$ for a fixed $k$ and how to calculate candidates.
If $|G_k| \le 2$, we just try to extend each border in $G_k$ to the right by using LCE queries on $T'$,
and if the extension reaches the right-end of $T'$,
we add its length to $\mathcal{C}$.
Note that we do not care about the periodicity of borders here.
Otherwise,
we compute $\alpha_k$ and $r_k$ by using LCE queries on $T'$.
Let $\tilde{b}_k$ be the longest element in $G_k$ whose length is at most $\alpha_k p_k-r_k$, if such a border exists.
If $\tilde{b}_k$ is defined, we check whether $\tilde{b}_kR$ is a border of $T'$ or not, again by using an LCE query on $T'$.
If $\tilde{b}_kR$ is a border of $T'$, we add its length to $\mathcal{C}$.
Also, we similarly check whether $b_k^{\min}R$ is a border of $T'$, and if so, add its length to $\mathcal{C}$,
where $b_k^{\min}$ is the shortest element in $G_k$, which may be non-periodic.

\paragraph*{Correctness.}

If a group $G_k$ contains at least three elements,
the borders in $G_k$ except for the shortest one are periodic (Proposition~\ref{prop:group}).
Namely, any non-periodic border of $Lw$ is either
an element of a group whose size is at most two, or
the shortest element of a group whose size is at least three.
Both cases are completely taken care of by our algorithm.
For periodic borders,
it is sufficient to check the longest border $\tilde{b}_k$ of $Lw$ whose length is at most $\alpha_{k}p_{k} - r_{k}$
for each group $G_k$ by Corollary~\ref{cor:bord}.
Therefore, the length of the border of $T'$ must belong to the candidates-set $\mathcal{C}$ obtained at the end of our algorithm.

\paragraph*{Running Time.}
Given a query $\phi(i, j, w)$, we can obtain the border array $\B_{Lw}$ and the border-group array $\BG_{Lw}$ in $O(\beta(n, |w|))$ time
from $\B_T[1..|L|]$ and $\BG_{T}[1.. |L|]$~(Lemmas~\ref{lem:border} and \ref{lem:bordergroup}).
Thus, by using those arrays,
while we scan the groups $G_1, \ldots, G_m$,
we can
determine whether the current group $G_k$ has at least three elements or not, and
compute the first term and the common difference of the arithmetic progression representing the current group $G_k$
both in constant time.
All the other operations consist of LCE queries on $T'$ and basic arithmetic operations, which can be done in constant time.
Finally, we choose the maximum from $\mathcal{C}$, which can be done $O(|\mathcal{C}|)$ time.
Since we add at most two elements to $\mathcal{C}$ when we process each group,
the size of $\mathcal{C}$ is in $O(m)$.
Thus the total running time is in $O(\beta(n, |w|) + |w| + \log n + m) \subseteq O(\beta(n, |w|) + |w| + \log n)$ since $m \in O(\log n)$.

To summarize this section, we obtain the following theorem.

\begin{theorem}\label{thm:border}
  The longest border after-edit query can be answered in $O(\beta(n, \ell) + \ell + \log n)$ time
  after $O(n)$-time preprocessing,
  where $\ell$ is the length of the string to be inserted or substituted specified in the query.
\end{theorem}
\section{Shortest Cover After Edit} \label{sec:cover}

This section proposes an algorithm to solve the {\SCAE} problem.
Firstly, we give additional notations and tools.
For a string $S$ and an integer $k$ with $1 \le k \le |S|$,
$\range(S, k)$ denotes the largest integer $r$ such that
$S[1.. k]$ can cover $S[1.. r]$.
Next, we give definitions of two arrays $\C(T)$ and $\R(T)$ introduced in~\cite{Breslauer92}.
The former $\C(T)$ is called the \emph{cover array} and stores the length of the cover of each prefix of $T$,
i.e., $\C(T)[k] = |\cover(T[1.. k])|$ for each $k$ with $1 \le k \le n$.
For convenience, let $\C(T)[0] = 0$.
The latter $\R(T)$ is called the \emph{range array} that stores the values of range function only for \emph{superprimitive prefixes} of $T$, i.e.,
for each $k$,
$\R(T)[k] = \range(T, k)$ if $\cover(T[1.. k]) = T[1.. k]$, and
otherwise $\R(T)[k] = 0$, meaning undefined.
Cover array and range array can be computed in $O(n)$ time given $T$ in an online manner~\cite{Breslauer92}.
In describing our algorithm, we use the next lemma:
\begin{lemma} \label{lem:range_array}
  Assume that we already have data structure $\mathcal{D}_T$
  consisting of
  the IPM data structure on $T$ of Lemma~\ref{lem:IPM},
  border array $\B(T)$,
  cover array $\C(T)$,
  range array $\R(T)$, and
  an array $\Rstar$ of size $n$ initialized with $\boldsymbol{0}$.
  Given a query $\phi(i, j, w)$,
  we can enhance $\mathcal{D}_T$ in $O(\beta(n, |w|) + |w|)$ time
  so that we can obtain
  $\cover((Lw)[1..k])$ for any $k$ with $1 \le k \le |Lw|$ 
  and
  $\range(Lw, k')$ for any $k'$ such that $1 \le k' \le |L|$ and $\cover(L[1.. k']) = L[1.. k']$
  in $O(1)$ time.
\end{lemma}
\begin{proof}
To prove the lemma, we first review Breslauer's algorithm~\cite{Breslauer92}
that computes $\C(T)$ and $\R(T)$ for a given string $T$ in an online manner~(Algorithm~\ref{alg:coverarray}).
\begin{algorithm}[tbh]
  \caption{Algorithm to compute $\C(T)$ proposed in~\cite{Breslauer92}}\label{alg:coverarray}
  \begin{algorithmic}[1]
    \Require{The border array $\B(T)$ of string $T$, and two arrays $C[0..n] = R[0..n] = \boldsymbol{0}$.}
    \Ensure{$C = \C(T)$ and $R = \R(T)$}
    \State $\idx \leftarrow 1$
    \While{$\idx \leq n$}
    \State $\clen \leftarrow C[\B(T)[\idx]]$ \Comment{$\clen < \idx$ always holds.}
    \If{$\clen > 0$ and $R[\clen] \geq \idx-\clen$}
    \State $C[\idx] \leftarrow \clen$
    \State $R[\clen] \leftarrow \idx$ \Comment{When $T[1.. \idx]$ is \emph{not} superprimitive, $R[\clen]$ is updated to $\idx$.}
    \Else
    \State $C[\idx] \leftarrow \idx$
    \State $R[\idx] \leftarrow \idx$ \Comment{When $T[1..\idx]$ is superprimitive, $R[\idx]$ is newly defined.}
    \EndIf
    \State $\idx \leftarrow \idx+1$
    \EndWhile
  \end{algorithmic}
\end{algorithm}
%

By Lemma~\ref{lem:border}, we can compute $\B(Lw)$ in $O(\beta(n, |w|))$ time if $\B(L)$ and $w$ are given.
Since Algorithm~\ref{alg:coverarray} runs in an online manner,
if we have $\C(L)$ and $\R(L)$ in addition to $\B(Lw)$, then it is easy to obtain $\C(Lw)$ and $\R(Lw)$
in $O(|w|)$ time
by running Algorithm~\ref{alg:coverarray} starting from the $(|L|+1)$-th iteration.
However, we only have $\C(T)$ and $\R(T)$, not $\C(L)$ and $\R(L)$.

The proof idea is to simulate $\C(L\cdot w[1.. t-1])$ and $\R(L\cdot w[1.. t-1])$
while iterating the while-loop of Algorithm~\ref{alg:coverarray} from $t = 1~(\idx = |L|+1)$ to $t = |w|~(\idx = |L| + |w|)$.
Note that in each $t$-th iteration of the algorithm,
only the first at most $\idx - 1 = |L| + t-1$ values of $C$ and $R$ may be referred since $\clen < \idx$ holds at line~$3$ of Algorithm~\ref{alg:coverarray}.
In the following, 
we show how to simulate the arrays
$\C(L\cdot w[1.. t-1])$ and $\R(L\cdot w[1.. t-1])$ for each $t$-th iteration
in an inductive way by looking at Algorithm~\ref{alg:coverarray}.

Firstly, to show the base-case $t = 1$, we consider the relations between the arrays of $T$ and those of $L$.
By the definition of the cover array,
$\C(L) = \C(T)[1.. |L|]$ holds since $L$ is a prefix of $T$.
Hence, we do not need any data structure other than $\C(T)$ to simulate $\C(L)$.
On the other hand, $\R(L)$ is not necessarily identical to $\R(T)[1.. |L|]$.
Specifically, since $L$ is a prefix of $T$,
$\R(L)[k] = \R(T)[k]$ holds if and only if $\R(T)[k] \le |L|$ holds
for every $k$ with $1 \le k \le |L|$.
Let $k$ be an integer such that $1 \le k \le |L|$ and $T[1.. k]$ is superprimitive.
If $\R(T)[k] > |L|$, we compute the value of $\R(L)[k]$ on demand.
We use the next claim.
\begin{claim} \label{claim:ipm}
  When $\R(T)[k] > |L|$,
  the rightmost occurrence of $T[1.. k]$ within $T[1.. |L|]$ must cover the position $|L|-k+1$,
  and thus, it occurs within $T[|L|-2k+2.. |L|]$ of length $2k-1$.
\end{claim}
Thus, if $\R(T)[k] > |L|$, we can obtain the value of $\R(L)[k]$
by answering the internal pattern matching query for two factors $u = T[1.. k]$ and $v = T[|L|-2k+2.. |L|]$ of $T$.
More precisely, $\R(L)[k] = s + k - 1$
where $s$ is the rightmost occurrence of $u$ in $v$,
which can be obtained by the internal pattern matching query on $T$ in constant time~(Lemma~\ref{lem:IPM}).
Hence, we can simulate $\C(L)$ and $\R(L)$ before the first iteration (i.e., $\idx = |L|+1$) of our algorithm.
By using this fact,
we can compute the values of $\C(L\cdot w[1])[\idx]$ and
either $\R(L\cdot w[1])[\clen]$ or $\R(L\cdot w[1])[\idx]$ correctly
(the correctness is due to~\cite{Breslauer92}).
Also, for an invariant of subsequent iterations,
we update $\Rstar[k] \leftarrow \R(L\cdot w[1])[k]$
where $k$ is either $\clen$ or $\idx$, which depends on the branch of the \textbf{if} statement.
Then, at the end of the first iteration, we can simulate $\C(L\cdot w[1])$
since we can already simulate $\C(L\cdot w[1])[1..|L|] = \C(L)$
and we have computed the last element of $\C(L\cdot w[1])$.
Next, we consider the range array.
Let $\mathit{last} = |L|+1$.
If $\Rstar[\mathit{last}] \ne 0$, then $\R(L\cdot w[1])[\mathit{last}] = \mathit{last}$ holds
because line~$9$ was executed.
Otherwise, $\R(L\cdot w[1])[\mathit{last}] = 0$ since $L\cdot w[1]$ is not superprimitive.
Next, let $k$ be an integer with $1 \le k < \mathit{last}$.
According to Algorithm~\ref{alg:coverarray},
the only candidate for the different element
between $\R(L)$ and $\R(L\cdot w[1])[1.. |L|]$ is the $\clen$-th element.
Thus, if $\Rstar[k] \ne 0$, then $k = \clen$ and $\R(L\cdot w[1])[\clen] = \Rstar[\clen]$ hold
since the $\clen$-th element of $\R(L)$ has been updated and $\Rstar[\clen]$ stores the updated value.
Otherwise, $\R(L\cdot w[1])[k]$ has never been updated, i.e., $\R(L\cdot w[1])[k] = \R(L)[k]$, which can be simulated as mentioned above.
Therefore, we can simulate $\R(L\cdot w[1])$ at the end of the first iteration.

Next, we consider the $(\iota+1)$-th iteration for some $\iota$ with $1 \le \iota \le |w|-1$.
Just after the $\iota$-th iteration (equivalently, just before the $(\iota+1)$-th iteration),
we assume that
(1) we can simulate $\C(L\cdot w[1.. \iota])$,
(2) we can simulate $\R(L\cdot w[1.. \iota])$, and
(3) for every position $k$ such that we have updated $R[k]$ in a previous iteration,
$\Rstar[k]$ stores the correct value of $\R(L\cdot w[1.. \iota])[k]$, which is a non-zero value.
Now, we show the three above invariants hold just after the $(\iota+1)$-th iteration.
Note that $\idx = |L| + \iota + 1$ in this step.
By the assumptions, we can compute $\C(L\cdot w[1.. \iota])[\B(L\cdot w[1.. \iota+1])[\idx]$ at line~$3$
and $\R(L\cdot w[1.. \iota])[\clen]$ at line~$4$ by simulating the arrays.
Also, the values of $\C(L\cdot w[1.. \iota+1])[\idx]$
and either $\R(L\cdot w[1.. \iota+1])[\clen]$ or $\R(L\cdot w[1.. \iota+1])[\idx]$ are computed correctly.
In addition to the original procedures, we update $\Rstar[k] \leftarrow \R(L\cdot w[1.. \iota+1])[k]$ appropriately, where $k$ is either $\clen$ or $\idx$,
and then, the third invariant holds at the end of the $(\iota+1)$-th iteration.
Further, similar to the first iteration,
we can simulate $\C(L\cdot w[1.. \iota+1])$, i.e., the first invariant holds.
Lastly, we show that we can simulate $\R(L\cdot w[1.. \iota+1])$ (holding the second invariant).
The proof idea is the same as the base-case.
Let $\mathit{last_{\iota}} = |L|+\iota+1$.
If $\Rstar[\mathit{last_{\iota}}] \ne 0$, then $\R(L\cdot w[1.. \iota+1])[\mathit{last_{\iota}}] = \mathit{last_{\iota}}$ holds
because $\Rstar[\mathit{last_{\iota}}]$ cannot be accessed in any previous iteration, and thus line~$9$ was executed in this iteration.
Otherwise, $\R(L\cdot w[1.. \iota+1])[\mathit{last_{\iota}}] = 0$ since $L\cdot w[1.. \iota+1]$ is not superprimitive.
Next, let $k_{\iota}$ be an integer with $1 \le k_{\iota} < \mathit{last_{\iota}}$.
According to Algorithm~\ref{alg:coverarray},
the only candidate for the different element
between $\R(L\cdot w[1.. \iota])$ and $\R(L\cdot w[1.. \iota+1])[1.. |L|+\iota]$ is the $\clen$-th element.
Thus, same as in the base-case,
$\R(L\cdot w[1.. \iota+1])[k_{\iota}] = \Rstar[k_{\iota}]$ holds
if $\Rstar[k_{\iota}] \ne 0$, and
$\R(L\cdot w[1.. \iota+1])[k_{\iota}] = \R(L\cdot w[1.. \iota])[k_{\iota}]$ holds otherwise.
Therefore, we can simulate $\R(L\cdot w[1.. \iota+1])$ at the end of the $(\iota+1)$-th iteration.

To summarize, running the above algorithm yields a data structure that can simulate $\C(Lw)$ and $\R(Lw)$.
This concludes the proof of Lemma~\ref{lem:range_array}.
\end{proof}
Note that we can prepare the input $\mathcal{D}_T$ of Lemma~\ref{lem:range_array} in $O(n)$ time for a given $T$.
A complete pseudocode is shown in Algorithm~\ref{alg:ours}.
\begin{algorithm}[ht]
  \caption{Algorithm to compute data structures which can simulate $\C(Lw)$ and $\R(Lw)$}\label{alg:ours}
  \begin{algorithmic}[1]
    \Require{$\B(T)$, $\C(T)$, $\R(T)$, $\Rstar[1.. n] = \boldsymbol{0}$, and $\phi(i, j, w)$.}
    \Ensure{(i) $C_w[1.. |w|] = \C(Lw)[|L|+1.. |Lw|]$ and (ii) $\Rstar[k] = \R(Lw)[k]$ if $1 \le k \le |Lw|$ and $\R(T)[k] \ne \R(Lw)[k] > |L|$, and $\Rstar[k] = 0$ otherwise.}
    \State $\idx \leftarrow |L|+1$ \Comment{Starting from $(|L|+1)$-th position.}
    \While{$\idx \leq |Lw|$}
    \State $\clen \leftarrow \C(T)[\B(T)[\idx]]$
    \If {$\clen > 0$} \Comment{$T[1.. \clen]$ is superprimitive.}
    \If {$\Rstar[\clen] \ne 0$}
    \State $r \leftarrow \Rstar[\clen]$
    \ElsIf {$\R(T)[\clen] \le |L|$}
    \State $r \leftarrow \R(T)[\clen]$
    \Else \Comment{$\Rstar[\clen] = 0$ and $\R(T)[\clen] > |L|$}
    \State $r \leftarrow \mathsf{rightend}(T[1.. \clen], T[|L|-2\clen+2.. |L|])$
    \EndIf
    \If{$r \geq \idx-\clen$} \Comment{$r = \R(L\cdot w[1.. \idx-|L|])[\clen]$}
    \State $C_w[\idx-|L|] \leftarrow \clen$
    \State $\Rstar[\clen] \leftarrow \idx$ \Comment{$(Lw)[1.. \idx]$ is \emph{not} superprimitive.}
    \State $\idx \leftarrow \idx+1$
    \State {\bf continue} \Comment{Go to the next iteration.}
    \EndIf
    \EndIf
    \State $C_w[\idx-|L|] \leftarrow \idx$
    \State $\Rstar[\idx] \leftarrow \idx$ \Comment{$(Lw)[1..\idx]$ is superprimitive.}
    \State $\idx \leftarrow \idx+1$
    \EndWhile
  \end{algorithmic}
\end{algorithm}

\paragraph*{Overview of Our Algorithm for {\SCAE} Queries}

To compute the cover of $T'$,
we first run the {\LBAE} algorithm of Section~\ref{sec:border}.
Then, there are two cases:
(i)  The \emph{non-periodic case}, where the length of $\border(T')$ is smaller than $|T'|/2$, or
(ii) the \emph{periodic case}, the other case.

\subsection{Non-periodic Case}
Let $b = \border(T')$ and $c = \cover(b)$.
By the first statement of Lemma~\ref{lem:cover_properties},
$\cover(T') = \cover(\border(T'))$
if $\cover(\border(T'))$ can cover $T'$, and
$\cover(T') = T'$ otherwise.
In the following, we consider how to determine whether 
$c = \cover(\border(T'))$ is a cover of $T'$ or not.

Let $s$ be the maximum length of the prefix of $Lw$ that $c$ can cover.
Further let $t$ be the maximum length of the suffix of $wR$ that $c$ can cover if $|c| \le |wR|$,
and $t = |c|$ otherwise.
By Lemma~\ref{lem:range_array}, the values of $s$ and $t$ can be obtained in $O(\beta(n, |w|) + |w|)$ time
by computing values of $\range(Lw, |c|)$ and $\range((wR)^R, |c|)$
since $c = \cover(b)$ is superprimitive~(by the second statement of Lemma~\ref{lem:cover_properties}),
where $(wR)^R$ denotes the reversal of $wR$.
If $s + t \ge |T'|$ then $c$ is a cover of $T'$.
Thus, $\cover(T') = c$, and the algorithm is terminated.

We consider the other case, where  $s + t < |T'|$.
The inequality $s + t < |T'|$ means that
the occurrences of $c$ within $Lw$ or $wR$ cannot cover the middle factor $T'[s+1..|T'|-t]$ of $T'$.
Thus, if $c$ is a cover of $T'$ when $s+t<|T'|$,
then $c$ must have an occurrence that starts in $L$ and ends in $R$.
Such an occurrence can be written as a concatenation of
some border of $Lw$ which is longer than $w$ and some (non-empty) prefix of $R$.
Similar to the method in Section~\ref{subsec:longborder},
we group the borders of $Lw$ using their periods and process them for each group.
Again, let $G_1, \ldots, G_m$ be the groups 
sorted in descending order of their smallest periods.

Let us fix a group $G_k$ arbitrarily.
If $|G_k| \le 2$, we simply try to extend each border in $G_k$ to the right by using LCE queries.
Now we use the following claim:
\begin{claim}
  For a border $z$ of $Lw$ with $|z| > |w|$,
  the value of $\lce_{T'}(|z|+1, |Lw|+1)$ can be computed in constant time
  by using the LCE data structure of Lemma~\ref{lem:lce} on $T$.
\end{claim}
This claim can be proven by similar arguments as in the first paragraph of Section~\ref{subsec:longborder}.
Thus, the case of $|G_k| \le 2$ can be processed in constant time.
If $|G_k| > 2$, we use the period $p_k$ of borders in $G_k$.
Let $\alpha_k$ be the exponent of the longest prefix of $T'$ with period $p_k$.
Further let $r_{k} = \lce_{T'}(|T'|-|R|-p_k +1,|T'|-|R|+1)$.
Note that $\alpha_kp_k < |c|$ since $c$ is a prefix of $T'$ and is non-periodic.
See also Figure~\ref{fig:non-periodic}.
\begin{figure}[tb]
  \center{\includegraphics[width=0.5\linewidth]{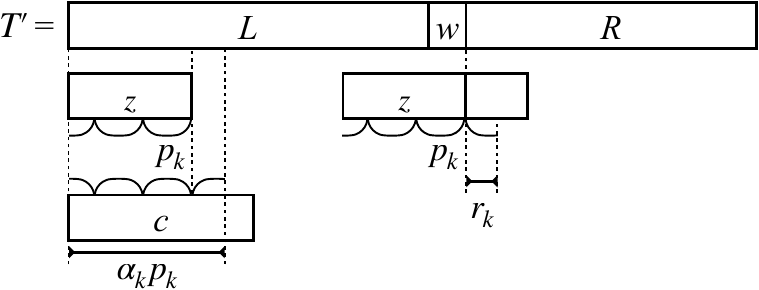}}
  \caption{Illustration for the non-periodic case.
    Here, $c$ is non-periodic, $z$ is some border of $Lw$, and $p_k$ is the period of $z$.
    If there is an occurrence of $c$ starting in $R$ and ending in $R$,
    then $|z| = \alpha_kp_k - r_k$ must hold since $c$ is non-periodic.
  } \label{fig:non-periodic}
\end{figure}
By using LCE queries on $T'$, $\alpha_k$ and $r_k$ can be computed in constant time.
For a border $z$ in $G_k$,
$T'[|z|+1.. |c|] = T'[|Lw|+1.. |Lw|+|c|-|z|]$ holds
only if $|z| = \alpha_kp_k - r_k$.
Thus, the only candidate for a border in $G_k$ which can be extended to the right enough
is of length exactly $\alpha_kp_k - r_k$ if it exists.
The existence of such a border can be determined in constant time
since the lengths of the borders in $G_k$ are represented as an arithmetic progression.
If such a border of length $\alpha_kp_k -r_k$ exists,
then we check whether it can be extended to the desired string $c$ by querying an LCE.
Therefore, the total computation time is $O(1)$ for a single group $G_k$,
and $O(\log n)$ time in total for all groups
since there are $O(\log n)$ groups.

To summarize, we can compute $\cover(T')$ in $O(\beta(n, |w|) + |w| + \log n)$ for the non-periodic case.

\subsection{Periodic Case}
In this case, $T'$ can be written as $(uv)^ku$ for some integer $k \ge 2$ and strings $u, v$ with $|uv| = \period(T')$
since $T'$ is periodic.
By the third statement of Lemma~\ref{lem:cover_properties}, $\cover(T') = \cover(uvu)$ holds
since $uvu$ is a cover of $T'$.
Thus, in the following, we focus on how to compute $\cover(uvu)$.
We further divide this case into two sub-cases depending on the relation between the lengths of $uvu$ and $Lw$.

If $|uvu| \le |Lw|$, then $uvu$ is a prefix of $Lw$.
Thus, by Lemma~\ref{lem:range_array}, $\cover(uvu) = \cover((Lw)[1.. |uvu|])$ can be computed in $O(\beta(n, |w|) + |w|)$ time.

If $|uvu| > |Lw|$, then $T' = uvuvu$ since $|uvu| > n/2$.
We call the factor $T[|uv|+1..|uvu|] = u$ the \emph{second occurrence of $u$}.
Also, since $|L| \ge |R| = |T'| - |Lw| > |T'| - |uvu| = |vu|$,
both $R$ and $L$ are longer than $uv$.
Thus $vu$ is a suffix of $R$ and $uv$ is a prefix of $L$.
Now let us consider the border of $uvu$. 
\begin{lemma}
  If the period of a string $T' = uvuvu$ is $|uv|$,
  then the border of $uvu$ is not longer than $|uv|$.
\end{lemma}
\begin{proof}
  If $uvu$ has a border that is longer than $|uv|$, $uvu$ has a period $p$ which is smaller than $|u|$. 
  Then the length-$p$ prefix of the second occurrence of $u$ repeats to the left and the right until it reaches both ends of $T'$~(see Figure~\ref{fig:uvuvu}).    
  This contradicts that $\period(T') = |uv|$. 
\end{proof}
\begin{figure}[tb]
  \center{\includegraphics[width=0.5\linewidth]{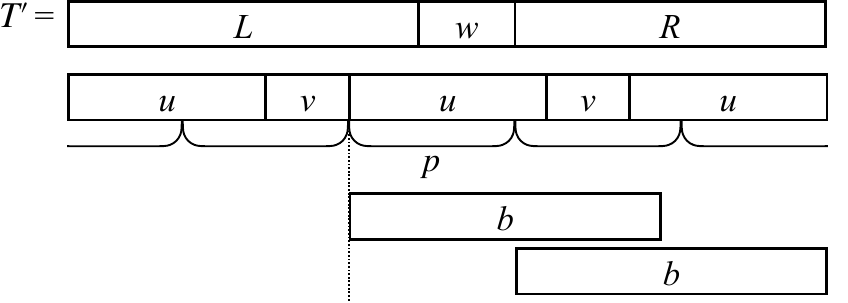}}
  \caption{Illustration for a contradiction if we assume that $uvu$ has a border which is longer than $|uv|$.
    Since $T' = uvuvu$, if $uvu$ has a period which is smaller than $|u|$ then $T'$ also has the same period.
  } \label{fig:uvuvu}
\end{figure}
Therefore, the border of $uvu$ is identical to the longest border of $T$ whose length is at most $|uv|$,
which can be obtained in constant time after $O(n)$-time preprocessing as in step~2 of Section~\ref{sec:border}.
By the first statement of Lemma~\ref{lem:cover_properties},
$\cover(uvu)$ is either $\cover(\border(uvu))$ or $uvu$.
Since $|\border(uvu)| \le |uv| < |Lw|$,  $\cover(\border(uvu))$ can be obtained in $O(\beta(n, |w|) + |w|)$ time by Lemma~\ref{lem:range_array}.
Let $x = \cover(\border(uvu))$.
Thanks to Lemma~\ref{lem:shortestcoverandtange} below, we do not have to scan $O(\log n)$ groups, unlike the non-periodic case.
\begin{lemma}\label{lem:shortestcoverandtange}
  When $|uvu| > |Lw|$,
  string $x = \cover(\border(uvu))$ covers $uvu$ if and only if
  $\range(Lw, |x|) \ge |uvu| - \max\{|u|, |x|\}$ holds.
\end{lemma}
\begin{proof}
  Let $r = \range(Lw, |x|)$.
  We divide the proof into three cases.
  \begin{description}
    \item[The case when $|x| \le |u|/2$:]
      In this case, $x$ is a border of $u$ and the occurrence of $x$ as the prefix of the second occurrence of $u$ ends within $Lw$.
      ($\implies$) If $x$ covers $uvu$, then $r \ge |uv| + |x| > |uv| = |uvu| - |u|$.
      ($\impliedby$) If $r \ge |uvu| - |u| = |uv|$ holds, then $x$ covers $uvx$ and $u$. Hence $x$ covers $uvu$~(see the left figure of Figure~\ref{fig:lastlemma}).
    \item[The case when $|u|/2 < |x| \le |u|$:]
      In this case, $x$ is a border of $u$ and its prefix-suffix occurrences in $u$ share the center position $\lceil |u|/2 \rceil$ of $u$.
      ($\implies$) Assume the contrary, i.e., $x$ covers $uvu$ and $r < |uv|$.
      Since $x$ covers $uvu$, there exists an occurrence of $x$ that covers position $r+1$.
      Also, since $r < |uv|$, the occurrence does not end within $Lw$.
      Thus, the occurrence must cover the center position $\lceil |u|/2 \rceil$ of the second occurrence of $u$.
      Now, there are three distinct occurrences of $x$ that cover the same position $\lceil |u|/2 \rceil$,
      however, it contradicts that $x = \cover(\border(uvu))$ is non-periodic (the second statement of Lemma~\ref{lem:cover_properties}).
      ($\impliedby$) Similar to the previous case, if $r \ge |uvu| - |u| = |uv|$ holds, then $x$ covers $uvx$ and $u$. Hence $x$ covers $uvu$.
    \item[The case when $|x| > |u|$:]
      Let $s = |uvu| - |x|$.
      In this case,
      $x$ occurs at positions $s+1$ and $|uv|+1$.
      Thus, the occurrences share position $|uv|+1$, which is the first position of the second occurrence of $u$~(see the right figure of Figure~\ref{fig:lastlemma}).
      ($\implies$) Assume the contrary, i.e., $x$ covers $uvu$ and $r < s$.
      Similar to the previous case,
      there must be an occurrence of $x$ such that
      the occurrence covers position $r+1$ and does not end within $Lw$.
      Then, there are three distinct occurrences of $x$ that cover the same position $|uv|+1$,
      which leads to a contradiction with the fact that $x$ is non-periodic.
      ($\impliedby$) This statement is trivial by the definitions and the conditions.
  \end{description}
\end{proof}
\begin{figure}[tb]
  \center{\includegraphics[width=\linewidth]{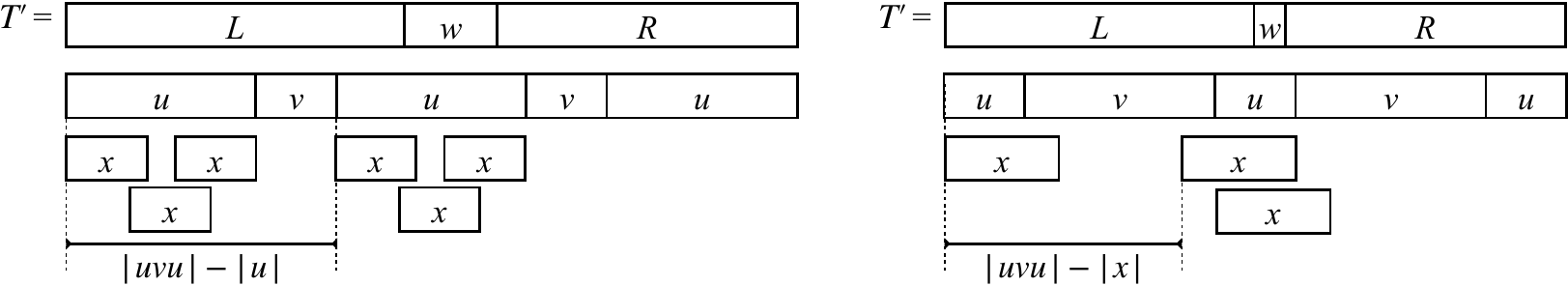}}
  \caption{Left: Illustration for the case $|x| \le |u|$.
    Right: Illustration for the case $|x| > |u|$.
  } \label{fig:lastlemma}
\end{figure}
Therefore, 
if $\range(Lw, |x|) \ge |uvu| - \max\{|u|, |x|\}$ then the cover of $uvu$ is $x$.
Otherwise, the cover of $uvu$ is $uvu$ itself.
Further, by Lemma~\ref{lem:range_array}, the value of $\range(Lw, |x|)$ can be obtained in $O(\beta(n, |w|) + |w|)$ time
since $x = \cover(\border(uvu))$ is superprimitive and $|x| \le |uv| < |Lw|$.

To summarize, we can compute $\cover(T')$ in $O(\beta(n, |w|) + |w|)$ for the periodic case.

Finally, we have shown the main theorem of this paper:
\begin{theorem}
  The shortest cover after-edit query can be answered in $O(\beta(n, \ell) + \ell + \log n)$ time after $O(n)$-time preprocessing,
  where $\ell$ is the length of the string to be inserted or substituted specified in the query.
\end{theorem}

\section{Conclusions and Discussions} \label{sec:conclusion}
In this paper, we introduced the problem of computing the longest border and the shortest cover in the after-edit model.
For each problem, we proposed a data structure that can be constructed in $O(n)$ time and can answer any query in
$O(\beta(n, \ell) + \ell + \log n) \subseteq O(\ell \log n)$ time
where $n$ is the length of the input string, and $\ell$ is the length of the string to be inserted or replaced.

As a direction for future research, we are interested in improving the running time.
For {\LBAE} queries, when the edit operation involves a single character,
an $O(\log(\min\{\log n, \sigma\}))$ query time can be achieved by exploiting the periodicity of the border:
we pre-compute all \emph{one-mismatch borders} and store the triple of mismatch position, mismatch character, and the mismatch border length for each mismatch border.
The number of such triples is in $O(n)$.
Furthermore, the number of triples for each position is $O(\min\{\log n, \sigma\})$ due to the periodicity of borders.
Thus, by employing a binary search on the triples for the query position, the query time is $O(\log(\min\{\log n, \sigma\}))$.
However, this algorithm stores all mismatch borders and cannot be straightforwardly extended to editing strings of length two or more.
It is an open question whether the query time of {\LBAE} and {\SCAE} queries
can be improved to $O(\ell + \log\log n)$ 
for an edit operation of length-$\ell$ string in general.
Furthermore, applying the results obtained in this paper to a more general problem setting, particularly the computation of borders/covers in a fully-dynamic string,
is a future work that needs further exploration.

\end{document}